\documentclass{amsart}
\usepackage{amsfonts}
\usepackage{amsmath}
\usepackage{graphicx}

\newtheorem{thm}{Theorem}[section]

\newtheorem{lem}[thm]{Lemma}

\theoremstyle{definition}
\newtheorem{defn}[thm]{Definition}
\theoremstyle{remark}
\newtheorem{rem}[thm]{Remark}
\numberwithin{equation}{section}

\input{tcilatex}

\begin{document}
\title[Stationary solutions of the Schr\"{o}dinger-Newton model ...]{Bounds
states of the Schr\"{o}dinger-Newton model in low dimensions}
\author{Joachim Stubbe* and Marc Vuffray*}
\address{*EPFL, IMB-FSB, Station 8, CH-1015 Lausanne, Switzerland}
\email{Joachim.Stubbe@epfl.ch,Marc.Vuffray@epfl.ch}
\thanks{}
\subjclass{35Q55, 35Q40, 47J10}
\keywords{Schr\"{o}dinger-Newton equations, nonlinear Schr\"{o}dinger
equation}
\date{11th December 2008}

\begin{abstract}
We prove the existence of quasi-stationary symmetric solutions with exactly $%
n\geq 0$ zeros and uniqueness for $n=0$ for the Schr\"{o}dinger-Newton model
in one dimension and in two dimensions along with an angular momentum $m\geq
0$. Our result is based on an analysis of the corresponding system of
second-order differential equations.
\end{abstract}

\maketitle

%


\section{Introduction}

We consider the Schr\"{o}dinger-Newton equations 
\begin{equation}
i\psi _{t}+\Delta \;\psi -\gamma \Phi \psi =0,\quad \Delta \;\Phi =|\psi
|^{2}  \label{SN}
\end{equation}%
on $\mathbb{R}^{d}$, $d\in \left\{ 1,2\right\} $, which is equivalent to the
nonlinear Schr\"{o}dinger equation 
\begin{equation}
i\psi _{t}+\Delta \psi +\gamma \left( G_{d}\left( \left\vert x\right\vert
\right) \ast |\psi |^{2}\right) \psi =0  \label{NLS}
\end{equation}%
where $G_{d}\left( \left\vert x\right\vert \right) $ denotes the Green's
function of the Laplacian on $\mathbb{R}^{d}$. Of physical interest are
solutions having finite energy $E$ and particle number (or charge) $N$ given
by 
\begin{equation}
E(\psi )=\frac{1}{2}\int_{\mathbb{R}^{d}}|\nabla {\psi }(x,t)|^{2}\;dx-\frac{%
\gamma }{4}\iint_{\mathbb{R}^{d}}G_{d}\left( \left\vert x-y\right\vert
\right) |\psi (x,t)|^{2}|{\psi }(y,t)|^{2}\;dxdy  \label{energy}
\end{equation}%
and 
\begin{equation}
N(\psi )=\int_{\mathbb{R}^{d}}|\psi (x,t)|^{2}\;dx,  \label{charge}
\end{equation}%
respectively.

In the physical and mathematical literature the Schr\"{o}dinger-Newton
system in three space dimensions has a long standing history. With $\gamma $
designating appropriate positive coupling constants it appeared first in
1954, then in 1976 and lastly in 1996 for describing the quantum mechanics
of a Polaron at rest by S. J. Pekar ~\cite{P54}, of an electron trapped in
its own hole by Ph. Choquard ~\cite{L77} and of self-gravitating matter by
R. Penrose ~\cite{P96}. In 1977, E.Lieb ~\cite{L77} showed the existence of
a unique spherically symmetric ground state in three space dimensions by
solving an appropriate minimization problem. This ground state solution $%
u_{\omega }(x)$, $\omega >0$ is a positive spherically symmetric strictly
decreasing function. In ~\cite{L80}, P.L. Lions proved the existence of
infinitely many distinct spherically symmetric solutions and claimed a proof
for the existence of anisotropic bound states in ~\cite{L86}.

While Lieb's existence proof can be easily extended to dimensions $d=4$ and $%
d=5$, the situation has been unclear for lower dimensions due to the lack of
positivity of the Coulomb interaction energy term. For the one-dimensional
and two-dimensional problem this difficulty has been overcome recently in ~%
\cite{CS2007} and \cite{S2008} and the existence of a unique spherically
symmetric ground state has been shown by solving a minimization problem.
Finally, existence and uniqueness of spherically symmetric ground states has
been shown in any dimension $d\leq 6$ in \cite{CSV2008} by employing a
shooting method for the associated system of ordinary differential
equations. In addition, for $d=1$, the existence of a unique antisymmetric
ground states (which is an eigenfunction of the parity-operator) has been
shown in \cite{CS2007} and, for $d=2$, the existence of unique angular
excitations (eigenfunction of the angular momentum operator) has been shown
in \cite{S2008}. However, in one and two dimensions, existence of higher
bound states remained open along with an angular momentum in the last case
and so far only numerical studies are available indicating the existence of
excited states, see e.g. ~\cite{HMT2003}. Our main result proves the
existence of such solutions in the attractive case $\gamma >0$.

\subsection{The one-dimensional case}

We study the existence and uniqueness of quasi stationary solutions of the
form%
\begin{equation}
\psi (t,x)=\left( \frac{x}{\left\vert x\right\vert }\right) ^{p}\varphi
(\left\vert x\right\vert )e^{-i\omega t},\quad \underset{|x|\rightarrow
\infty }{\lim }\varphi (\left\vert x\right\vert )=0,  \label{eqref1d}
\end{equation}%
where $p=\left\{ 0,1\right\} $ for an odd or even function. For solutions of
the form \eqref{eqref1d} we have $\Phi (t,x)=v(|x|)$ and $\varphi (r),v(r)$
satisfy the following system of ordinary differential equations: 
\begin{equation}
\begin{split}
& \varphi ^{\prime \prime }=(\gamma v-\omega )\varphi \\
& v^{\prime \prime }=\varphi ^{2},\quad r\geq 0. \\
&
\end{split}
\label{ODE-1d}
\end{equation}%
We suppose that $v\left( 0\right) $ are finite and $v^{\prime }\left(
0\right) =0.$ In addition $\phi $ is subject to the initial condition 
\begin{align*}
\phi \left( 0\right) & \in 
\mathbb{R}
_{+}\text{, }\phi ^{\prime }\left( 0\right) =0\text{ \ \ \ \ if~}p=0 \\
\phi \left( 0\right) & =0\text{, \ }~\phi ^{\prime }\left( 0\right) \in 
\mathbb{R}
_{+}\text{ \ \ ~if~}p=1
\end{align*}%
The latter equation of (\ref{ODE-1d}) implies that $v^{\prime }\geq 0$, and
therefore for solutions $\phi $ vanishing at infinity we have $\omega
-\gamma v\left( 0\right) >0.$ Changing the variables and rescaling $u\left(
r\right) =Ar^{-p}~\phi \left( r/\sigma \right) $, $V\left( r\right) =B\left(
v\left( r/\sigma \right) -\omega /\gamma \right) +1$ with%
\begin{equation*}
\sigma =\sqrt{\omega -\gamma v\left( 0\right) },~~A=\frac{\sqrt{\gamma }}{%
\sigma ^{2}},~~B=\frac{\gamma }{\sigma ^{2}},
\end{equation*}%
we obtain the following system of equations:

\begin{equation}
\begin{split}
& u^{\prime \prime }+\frac{2p}{r}\;u^{\prime }=(V-1)u \\
& V^{\prime \prime }=u^{2}r^{2p} \\
&
\end{split}
\label{ODE-1d-re}
\end{equation}%
subject to the initial conditions 
\begin{equation}
u(0)=u_{0}\in \mathbb{R}^{+},\;u^{\prime }(0)=0,\;V(0)=0,V^{\prime }(0)=0.
\label{ini 1D}
\end{equation}

\subsection{The two-dimensional case}

We search solutions of (\ref{SN}) uniformly rotating with an angular
velocity $\Omega $. In the rotating frame the equation (\ref{SN}) becomes%
\begin{equation*}
i\psi _{t}+\Delta \;\psi -\gamma \Phi \psi -\Omega ~L\psi =0,\quad \Delta
\;\Phi =|\psi |^{2}
\end{equation*}%
where $L=-i\left( x\partial _{y}-y\partial _{x}\right) $, the orbital
angular momentum operator.

We study the existence and uniqueness of quasi stationary solutions in the
rotating frame of the form 
\begin{equation}
\psi (t,x)=\varphi (|x|)e^{-i\omega t+im\theta },\quad \underset{%
|x|\rightarrow \infty }{\lim }\varphi (|x|)=0.  \label{eqref2d}
\end{equation}%
If, in addition, 
\begin{equation*}
\varphi (|x|)\geq 0,
\end{equation*}%
we call this solution the ground states for a given orbital angular momentum 
$m$, the other one are called excited states. For solutions of the form %
\eqref{eqref2d} we have $\Phi (t,x)=v(|x|)$ and $\varphi (r),v(r)$ satisfy
the following system of ordinary differential equations: 
\begin{equation}
\begin{split}
& \varphi ^{\prime \prime }+\frac{1}{r}\;\varphi ^{\prime }-\frac{m^{2}}{%
r^{2}}\varphi =(\gamma v-(\omega -\Omega m))\varphi \\
& v^{\prime \prime }+\frac{1}{r}\;v^{\prime }=\varphi ^{2},\quad r\geq 0. \\
&
\end{split}
\label{ODE-2d}
\end{equation}%
We suppose that $v\left( 0\right) $, $\phi \left( 0\right) $ are finite and $%
v^{\prime }\left( 0\right) =\phi ^{\prime }\left( 0\right) =0.$ The latter
equation implies that $v^{\prime }\geq 0$, and therefore for solutions $\phi 
$ vanishing at infinity we have $\omega -\Omega m-\gamma v\left( 0\right)
>0. $ Changing the variables and rescaling $u\left( r\right) =r^{-m}A~\phi
\left( r/\sigma \right) $, $V\left( r\right) =B\left( v\left( r/\sigma
\right) -\left( \omega -\Omega m\right) /\gamma \right) +1$ with%
\begin{equation*}
\sigma =\sqrt{\omega -\Omega m-\gamma v\left( 0\right) },~~A=\frac{\sqrt{%
\gamma }}{\sigma ^{2}},~~B=\frac{\gamma }{\sigma ^{2}},
\end{equation*}%
we obtain the following system of equations:

\begin{equation}
\begin{split}
& u^{\prime \prime }+\frac{2\left( m+1\right) }{r}\;u^{\prime }=(V-1)u \\
& V^{\prime \prime }+\frac{2}{r}\;V^{\prime }=u^{2}r^{2m} \\
&
\end{split}
\label{ODE-2d-re}
\end{equation}%
subject to the initial conditions (\ref{ini 1D}).

\subsection{\protect\bigskip The general case}

Equations (\ref{ODE-1d-re}), (\ref{ODE-2d-re}) are particular cases of the
following initial value problem.%
\begin{equation}
\begin{split}
& u^{\prime \prime }+\frac{2m+d-1}{r}\;u^{\prime }=(V-1)u \\
& V^{\prime \prime }+\frac{d-1}{r}\;V^{\prime }=u^{2}r^{2m} \\
&
\end{split}
\label{ODE-d-re}
\end{equation}%
subject to the initial conditions 
\begin{equation}
u(0)=u_{0}\in \mathbb{R}^{+},\;u^{\prime }(0)=0,\;V(0)=0,V^{\prime }(0)=0.
\label{ini}
\end{equation}%
Here, $m\geq 0$ may be regarded as a continuous parameter. By analyzing the
solutions of the above initial value problem we shall prove the following
result about the existence and uniqueness of ground states and the existence
of the excited states:

\begin{thm}
\label{main-theorem} For $d=1,2$, any $m\geq 0$ and $n\geq 0$ the system %
\eqref{ODE-d-re} subject to the initial conditions \eqref{ini} admits a
solution $(u_{m,n},V_{m,n})$ such that $u_{m,n}$ has exactly $n$ zeros and 
\begin{equation}
\underset{r\rightarrow \infty }{\lim }u_{m,n}(r)=0.  \label{decay}
\end{equation}%
Moreover, the ground states $u_{m,0}$ are unique and satisfies $u_{m,0}(r)>0$%
, $u_{m,0}^{\prime }(r)<0$ on $]0,\infty \lbrack $.
\end{thm}

To prove the main result we use an extension of the shooting method we have
used in \cite{CSV2008}. Shooting methods have been successfully applied to
existence and uniqueness of solutions in boundary value problems for second
order nonlinear differential equations ~\cite{BLP81},~\cite{PS83},~\cite%
{AP86},~\cite{K89} . Our paper is organized as follows: In Section 2 we
employ a shooting method to prove the existence of ground states and by
induction the existence of all excited states (theorem \ref{existence-gs}).
In Section 3 study their decay properties to prove uniqueness by analyzing
the Wronskian of solutions (theorem \ref{uniqueness-in-G}).

\section{Existence of bound states}

We begin our study with the discussion of some general properties of
solutions of \eqref{ODE-d-re} with initial values \eqref{ini}. Standard
results will guarantee local existence and uniqueness of solutions, their
continuous dependence on the initial values as well as on the parameter $m$
and their regularity. As a consequence of local existence and uniqueness
solutions can only have simple zeros. We shall frequently apply these
properties in the sequel as well as the following integral equations for $%
u^{\prime }$ and $V^{\prime }$:

\begin{equation}
\begin{split}
& u^{\prime }(r)=\frac{1}{r^{2m+d-1}}\int_{0}^{r}(V(s)-1)u(s)s^{2m+d-1}\;ds
\\
& V^{\prime }(r)=\frac{1}{r^{d-1}}\int_{0}^{r}u^{2}(s)s^{2m+d-1}\;ds. \\
&
\end{split}
\label{int-eq}
\end{equation}%
Since $u_{0}>0$, we see from integration of the equation \ref{int-eq} that $%
V\left( r\right) $ is increasing, since $d\leq 2$, and goes to infinity.
Thus for any $u_{0}$ there exist a unique $a>0$, such that $V\left( a\right)
=1$. This is indeed a crucial point to restrict ourselves to dimensions $%
\leq 2$.

For the initial condition $u_{0}>0$ of the solution $(u,V)$ for an $m\geq 0$%
, we consider the following sets:

\begin{defn}
\begin{equation}
\mathcal{N}_{m,n}:=\left\{ u_{0}\in 
\mathbb{R}
_{+}:u\text{ has exactly }n\text{ zeros}\right\}  \label{Ni set}
\end{equation}%
\begin{equation}
\mathcal{N}_{m,n}^{0}:=\left\{ u_{0}\in \mathcal{N}_{m,n}:\lim_{r\rightarrow
\infty }u\left( r\right) =0\right\}  \label{Ni-zero set}
\end{equation}%
\begin{equation}
\mathcal{N}_{m,n}^{\infty }:=\mathcal{N}_{m,n}\setminus \mathcal{N}_{m,n}^{0}
\label{Ni-inf set}
\end{equation}%
\begin{equation}
\mathcal{N}_{m}:=\bigcup\limits_{n=1}^{\infty }\mathcal{N}_{m,n}
\label{N set}
\end{equation}
\end{defn}

Note that since $u^{\prime \prime }(0)=-u_{0}/\left( 2m+d\right) $ all
solutions start strictly decreasing.\ Therefore if $u$ has a first critical
point $r_{1}>0$ where $u>0$, then $V(r_{1})\geq 1$ and since $V$ is strictly
increasing (see \eqref{int-eq}) it follows again from \eqref{int-eq} that $%
u^{\prime }(r)>0$ for all $r>r_{1}$. Hence every solution $u_{0}\notin 
\mathcal{N}_{m,0}^{\infty }$ is strictly decreasing in the maximal interval $%
\left( 0,R\right) $ where $u>0$. Obviously the sets $\mathcal{N}%
_{m,n}^{\infty }$, $\mathcal{N}_{m,n}^{0}$ are mutually disjoint and we will
show that these sets forms a partition of $%
\mathbb{R}
_{+}$. First we prove that $u$ can oscillate only in a finite range and then
we show that $u$ has only a finite number of zero in this range.

\begin{lem}
\label{number of zeros in V>1}Let $a$ be such that $V\left( a\right) =1,$
then $u$ has at most one zero in the range $r>a$.

\begin{proof}
If $u$ has two zeros for $r>a,$ there exist a critical point $\overline{r}$
such that $u\left( \overline{r}\right) \neq 0$, $u\left( \overline{r}\right)
u^{\prime \prime }\left( \overline{r}\right) \leq 0$. Moreover from equation
(\ref{int-eq}) we have%
\begin{equation*}
u\left( \overline{r}\right) u^{\prime \prime }\left( \overline{r}\right)
=(V\left( \overline{r}\right) -1)u\left( \overline{r}\right) ^{2}>0
\end{equation*}%
which is the desired contradiction.
\end{proof}
\end{lem}

As consequence of the lemma \ref{number of zeros in V>1}, the number of
zeros of $u$ can only increase by one in a neighborhood of $u_{0}$.

\begin{lem}
\label{lemma of the open ball of nods}For any $u_{0}\in \mathcal{N}_{m,n}$,
there exist $\varepsilon >0$ such that $\mathcal{B}\left( u_{0},\varepsilon
\right) \in \mathcal{N}_{m,n}\cup \mathcal{N}_{m,n+1}\,,$ where $\mathcal{B}%
\left( u_{0},\varepsilon \right) $ is the open ball centered in $u_{0}$ with
radius $\varepsilon $.

\begin{proof}
We choose $R>\max \left( a,r_{n}\right) $ where $r_{n}$ is the last zero of $%
u$. Since $u$ has simple zeros it follows from the continuous dependence on
initial values that for every $\widetilde{u_{0}}$ in a open ball $\mathcal{B}%
\left( u_{0},\varepsilon \right) $, $\widetilde{u}$ has $n$ zeros in the
interval $\left[ 0,R\right] $ and $\widetilde{V}\left( R\right) >1$. Finally
from lemma \ref{number of zeros in V>1} $\widetilde{u}$ cannot have more
than $n+1$ zeros.
\end{proof}
\end{lem}

\begin{lem}
$u$ has a finite number of zeros.

\begin{proof}
It is sufficient to prove this assertion in the range $r\leq a$. We compare
equation (\ref{ODE-d-re}) to%
\begin{equation*}
\overline{u}^{\prime \prime }\left( r\right) +\frac{2m+d-1}{r}\overline{u}%
^{\prime }\left( r\right) +\overline{u}\left( r\right) =0,
\end{equation*}%
which admits the solution:%
\begin{equation*}
\overline{u}\left( r\right) =u_{0}\Gamma \left( m+\frac{d}{2}\right) \left( 
\frac{2}{r}\right) ^{\left( 2m+d-2\right) /2}J_{\left( 2m+d-2\right)
/2}\left( r\right) ,
\end{equation*}%
where $J_{\left( 2m+d-2\right) /2}$ is the Bessel function of first kind of
order $\left( 2m+d-2\right) /2$. By the Sturm's comparison theorem, see e.g. 
\cite{K89}, $u$ oscillate less than $\overline{u}$ \ in the range $r\leq a$.
Since $J_{\left( 2m+d-2\right) /2}$ has a finite number of zeros in this
range, the same properties also hold for $u$.
\end{proof}
\end{lem}

The reason of the notation $\mathcal{N}_{m,n}^{\infty }$ become from the
fact that the solutions $u$ in these sets goes to infinity as we will see in
the next lemma.

\begin{lem}
For every $u_{0}\in \mathcal{N}_{m,n}^{\infty }~$there exist $r_{c}\in
\left( 0,\infty \right] $ such that $\lim_{r\rightarrow r_{c}}\left\vert
u\left( r\right) \right\vert =\infty $.

\begin{proof}
We have seen in the lemma \ref{number of zeros in V>1} that, for large $r$,
solutions of \ref{ODE-d-re} have no zeros and, more precisely, are not
oscillatory. Then, supposing $\lim \inf_{r\rightarrow \infty }u\left(
r\right) >0$ or $\lim \sup_{r\rightarrow \infty }u\left( r\right) <0$ leads
to a contradiction. In the case $\lim \inf_{r\rightarrow \infty }u\left(
r\right) =l>0$, choosing $\overline{r}>a\ $sufficiently large and
integrating (\ref{int-eq}) gives:%
\begin{eqnarray*}
u^{\prime }\left( r\right) &=&\frac{1}{r^{2m+d-1}}\int_{\overline{r}%
}^{r}(V(s)-1)u(s)s^{2m+d-1}\ ds+\left( \frac{\overline{r}}{r}\right)
^{2m+d-1}u^{\prime }\left( \overline{r}\right) \\
&\geq &\frac{l~\left( V\left( \overline{r}\right) -1\right) }{2\left(
2m+d\right) }\left( r-\overline{r}\right) +\left( \frac{\overline{r}}{r}%
\right) ^{2m+1}u^{\prime }\left( \overline{r}\right) ~~~~~\forall r>%
\overline{r}.
\end{eqnarray*}%
The case $\lim \sup_{r\rightarrow \infty }u\left( r\right) <l$ can be proved
in the same manner.
\end{proof}
\end{lem}

\begin{remark}
Since as $u$ has only simple zeros and the preceeding lemma, we deduce from the
continuous dependence on initial values that $\mathcal{N}_{m,n}^{\infty }$, $%
\mathcal{N}_{m}\mathcal{\,}$are open sets.
\end{remark}

Our main result theorem \ref{main-theorem} states that each $\mathcal{N}%
_{m,n}^{0}$ contain at least on element. Obviously, the set of the ground
states $\mathcal{N}_{m,0}^{0}$ is nonempty if both $\mathcal{N}_{m}$ and $%
\mathcal{N}_{m,0}^{\infty }$ are nonempty. As we will show in the proof of
the main-theorem, the bound states are the infima of the $\mathcal{N}_{m,n}$
sets. Thus, a necessary condition for the existence of bound states for each 
$n$ is these infima are ever strictly positive. We will show these
properties in the two following lemmas.

\begin{lem}
\label{P-lemma} The set $\mathcal{N}_{m,0}^{\infty }$ is non empty.

\begin{proof}
We want to show that $u_{0}\in \mathcal{N}_{m,0}^{\infty }$ for $u_{0}~$%
sufficiently large. Suppose on the contrary that $u_{0}\in \mathcal{N}_{m}%
\mathcal{\cup N}_{m,0}^{0}$ for all $u_{0}>0$\ and denote $\left] 0,R_{0}%
\right[ \,\ $the maximal interval where $u>0\,\ $and $u^{\prime }<0$. Let $%
r_{0}=\sqrt{2(2m+d)}$. We consider the function%
\begin{equation*}
f\left( r\right) :=u\left( r\right) -u_{0}\left( 1-\frac{r^{2}}{r_{0}^{2}}%
\right) .
\end{equation*}%
It satisfied the differential equation%
\begin{equation}
f^{\prime \prime }\left( r\right) +\frac{2m+d-1}{r}f^{\prime }\left(
r\right) =V\left( r\right) u\left( r\right) +u_{0}-u\left( r\right) .
\label{equation for f}
\end{equation}%
Since $f^{\prime }\left( 0\right) =f\left( 0\right) =0$ and the right
hand-side of (\ref{equation for f}) we conclude that $f\left( r\right) \geq
0 $ on $\left[ 0,R_{0}\right[ $. Hence%
\begin{equation}
u\left( r\right) \geq u_{0}\left( 1-\frac{r^{2}}{r_{0}^{2}}\right) ~~~\text{%
on}~\left[ 0,R_{0}\right[  \label{lower bound u}
\end{equation}%
and $r_{0}\leq R_{0}$.Using that $u$ is decreasing and from the equation (%
\ref{int-eq}) for $V^{\prime }$ we obtain the bound%
\begin{equation*}
V^{\prime }\left( r\right) \geq \frac{u\left( r\right) ^{2}}{2m+d}r^{2m+1}~~~%
\text{on}~\left[ 0,R_{0}\right[ .
\end{equation*}%
Integrating this inequality and using again that $u^{\prime }<0$ yields the
following estimate for $V:$%
\begin{equation}
V\left( r\right) \geq \frac{u\left( r\right) ^{2}}{r_{0}^{2}\left(
m+1\right) }r^{2\left( m+1\right) }~~~\text{on}~\left[ 0,R_{0}\right[ .
\label{lower bound v}
\end{equation}%
We want to show that $u^{\prime }\left( r_{0}\right) >0$ provided $u_{0}$
sufficiently large which yields the desired contradiction. Using (\ref{lower
bound u}), (\ref{lower bound v}) in equation (\ref{int-eq}) leads to%
\begin{eqnarray*}
u^{\prime }\left( r_{0}\right) &=&\frac{1}{r_{0}^{2m+d-1}}%
\int_{0}^{r_{0}}\left( V\left( r\right) -1\right) u\left( r\right)
r^{2m+d-1}dr \\
&\geq &\frac{1}{r_{0}^{2m+d+1}\left( m+1\right) }\int_{0}^{r_{0}}u\left(
r\right) ^{3}r^{4m+d+1}dr-\frac{2u_{0}}{r_{0}} \\
&\geq &\frac{u_{0}^{3}}{r_{0}^{2m+d+1}\left( m+1\right) }\int_{0}^{r_{0}}%
\left( 1-\frac{r^{2}}{r_{0}^{2}}\right) ^{3}r^{4m+d+1}dr-\frac{2u_{0}}{r_{0}}
\\
&\geq &\frac{u_{0}}{r_{0}}\left( \frac{u_{0}^{2}}{m+1}r_{0}^{2\left(
m+1\right) }\int_{0}^{1}\left( 1-t^{2}\right) ^{3}t^{4m+d+1}dt-2\right) .
\end{eqnarray*}%
We conclude that $u^{\prime }\left( r_{0}\right) >0$ for $u_{0}$
sufficiently large which contradicts the assumption that $\mathcal{N}_{m}%
\mathcal{\cup N}_{m,0}^{0}=\left] 0,\infty \right[ $.
\end{proof}
\end{lem}

\begin{lem}
\label{lemma on the lower bound of nods}For any $N>0$, there exist $%
\widetilde{u_{0}}>0$ such that for each $u_{0}<\widetilde{u_{0}}$, the
corresponding solution has at least $N$ zeros i.e. $u_{0}\in \cup _{n\geq N}%
\mathcal{N}_{m,n}$

\begin{proof}
First we consider equation (\ref{ODE-d-re})\ for $u_{0}>0~$in the range $%
r\in \left[ 0,b\right] $, where $b$ is the value at which $V\left( b\right)
=1/2$. We introduce a Lyapunov function $F\left( r\right) $ defined by%
\begin{equation*}
F\left( r\right) =u^{2}\left( 1-V\right) +u^{\prime }{}^{2}.
\end{equation*}%
Then $F\left( 0\right) =u_{0}$ and 
\begin{equation*}
F^{\prime }\left( r\right) =-\frac{2\left( 2m+d-1\right) }{r}u^{\prime
}\left( r\right) ^{2}-V^{\prime }\left( r\right) u\left( r\right) ^{2}\leq 0.
\end{equation*}%
Hence $F\left( 0\right) \geq F\left( r\right) $ for all $r\in \left[ 0,b%
\right] $ and, using $V\left( r\right) \leq 1$ we therefore have%
\begin{equation*}
2u_{0}^{2}\geq u\left( r\right) ^{2}.
\end{equation*}%
Then from the equation \ref{int-eq} we find%
\begin{equation*}
\frac{1}{2}=V\left( b\right) =\int_{0}^{b}\int_{0}^{t}\frac{u\left( s\right)
^{2}s^{2m+d-1}}{t}dsdt\leq \frac{u_{0}^{2}}{\left( m+1\right) \left(
2m+d\right) }b^{2\left( m+1\right) },
\end{equation*}%
hence%
\begin{equation*}
b\geq \left( \frac{\left( m+1\right) \left( 2m+d\right) }{2u_{0}^{2}}\right)
^{\frac{1}{2\left( m+1\right) }}.
\end{equation*}%
We can now compare the main equation (\ref{ODE-d-re})\ to the following%
\begin{equation*}
\overline{u}^{\prime \prime }\left( r\right) +\frac{2m+d-1}{r}\overline{u}%
\left( r\right) +\frac{1}{2}\overline{u}\left( r\right) =0
\end{equation*}%
which admit the solution:%
\begin{equation*}
\overline{u}\left( r\right) =u_{0}\Gamma \left( m+\frac{d}{2}\right) \left( 
\frac{2\sqrt{2}}{r}\right) ^{\left( 2m+d-2\right) /2}J_{\left( 2m+d-2\right)
/2}\left( \frac{r}{\sqrt{2}}\right) ,
\end{equation*}%
where $J_{\left( 2m+d-2\right) /2}$ is the Bessel function of first kind. By
Sturm's comparison theorem, \cite{K89}, $u$ oscillates faster than $%
\overline{u}$ and has at least the same number of zero than $J_{\left(
2m+d-2\right) /2}\left( r/\sqrt{2}\right) $ in the range $0\leq r\leq \left( 
\frac{\left( m+1\right) \left( 2m+d\right) }{2u_{0}^{2}}\right) ^{\frac{1}{%
2\left( m+1\right) }}$. Therefore we can choose $u_{0}$ sufficiently small
such the number of zero of $u$ is greater or equal than $N$.
\end{proof}
\end{lem}

Hence we have proved by the preceding lemmas the existence of ground states
(see also \cite{S2008}). We are now in position to prove our main existence
results.

\begin{thm}
\label{existence-gs} For any $m\geq 0$ and $n\in 
\mathbb{N}
$, $\alpha _{m,n}:=\inf \mathcal{N}_{m,n}\in \mathcal{N}_{m,n}^{0}$. In
addition for every $m\geq 0$, the solution $u_{m,0}\in \mathcal{N}%
_{m,0}^{0}~ $satisfies $u_{m,0}(r)>0$, $u_{m,0}^{\prime }(r)<0$ on $%
]0,\infty \lbrack $.

\begin{proof}
We prove this first statement by induction with the following hypothesis:
For any $n\in 
\mathbb{N}
,$ $\alpha _{n}:=\inf \mathcal{N}_{m,n}\in \mathcal{N}_{m,n}^{0}$. We have
seen in the lemma \ref{lemma on the lower bound of nods} that $\alpha _{0}>0$%
. Since $\mathcal{N}_{m,0}^{\infty }$ and $\mathcal{N}_{m}$ are open sets, $%
\alpha _{0}\in \mathcal{N}_{m,0}^{0}$ and this result is true for $n=0$. By
hypothesis $\alpha _{n}\in \mathcal{N}_{m,n}^{0}$ and the lemma \ref{lemma
of the open ball of nods} applied at the point $\alpha _{n}$ states there
exist $\varepsilon >0$ such that $\mathcal{B}\left( \alpha _{n},\varepsilon
\right) \subset \mathcal{N}_{m,n}\cup \mathcal{N}_{m,n+1}$. Since $\alpha
_{n}$ is the infimum of $\mathcal{N}_{m,n}$ we have the following result%
\begin{equation*}
\mathcal{N}_{m,n+1}\neq \emptyset \text{ and }\alpha _{n+1}<\alpha _{n}
\end{equation*}%
and it follows from the lemma \ref{lemma on the lower bound of nods} that $%
\alpha _{n+1}>0$.

Now we suppose on the contrary that $\alpha _{n+1}\notin \mathcal{N}_{m,n+1}$%
. Since $\alpha _{n+1}<\alpha _{n}$, lemma \ref{lemma of the open ball of
nods} implies there exist an open ball around $\alpha _{n+1}~$which is not
in $\mathcal{N}_{n+1}$ in contradiction with the definition of $\alpha
_{n+1} $. Finally, since $\mathcal{N}_{m,n}^{\infty }$ is an open set, it
follows that $\alpha _{n+1}\in \mathcal{N}_{m,n}^{0}$.

It remains to prove the last part of the theorem. We have seen that every
solution $u\notin \mathcal{N}_{m,0}^{\infty }$ is strictly decreasing in the
maximal interval $\left( 0,R\right) $ where $u>0.$ Since the are only simple
zeros, we have\ in particular for the ground states: $u(r)>0$, $u^{\prime
}(r)<0$ on $]0,\infty \lbrack $.
\end{proof}
\end{thm}

\begin{rem}
We also prove with the preceding theorem that for all $m\geq 0$ the sequence 
$\left\{ \alpha _{n}\right\} _{n\in 
\mathbb{N}
}$of the infima is strictly decreasing as illustrate in the following figure.
\end{rem}

\begin{center}
\FRAME{itbpFU}{4.1719in}{2.7415in}{0in}{\Qcb{The curves represente the sets $%
\mathcal{N}_{m,n}$ for $n=0,1,2,3,4,5$ as a function of $m$ in dimension $d=2
$.}}{}{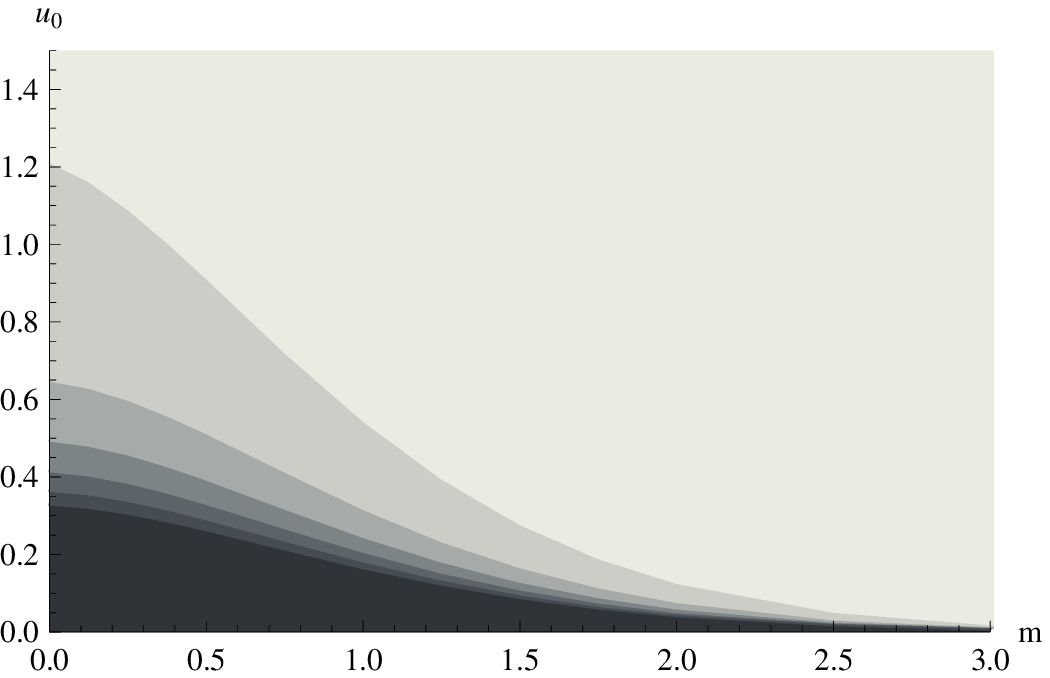}{\special{language "Scientific Word";type
"GRAPHIC";maintain-aspect-ratio TRUE;display "USEDEF";valid_file "F";width
4.1719in;height 2.7415in;depth 0in;original-width 4.1632in;original-height
2.7268in;cropleft "0";croptop "1";cropright "1";cropbottom "0";filename
'ill_ensemb_fct_m.eps';file-properties "XNPEU";}}
\end{center}

\section{Uniqueness of ground states}

In this section we prove that $\mathcal{N}_{m,0}^{0}$ has exactly one
element. First of all, we show that if $\mathcal{N}_{m,0}^{0}$ had more than
one element the corresponding solutions cannot cross. We restate the
no-crossing properties of \cite{CSV2008} in the following lemma.

\begin{lem}
\label{no-cross} Let $u_{2}(0)>u_{1}(0)>0$ and suppose that $%
u_{2}(r),u_{1}(r)$ exist on $[0,R]$ such that $u_{1}(r)\geq 0$ on $[0,R]$.
Then $u_{2}(r)>u_{1}(r)$ for all $r\in [0,R]$.
\end{lem}

\begin{proof}
We consider the Wronskian of $u_{1},u_{2}$ defined by 
\begin{equation}
w(r)=u_{2}^{\prime }(r)u_{1}(r)-u_{1}^{\prime }(r)u_{2}(r).
\label{Wronskian}
\end{equation}%
Then $w$ satisfies the differential equation 
\begin{equation}
w^{\prime }+\frac{2m+d-1}{r}\;w=(V_{2}-V_{1})u_{1}u_{2}.
\label{wronskian-ode}
\end{equation}

Suppose there is $\bar{r}\in \lbrack 0,R]$ such that $u_{2}(r)>u_{1}(r)$ on $%
[0,\bar{r}[$ and $u_{1}(\bar{r})=u_{2}(\bar{r})\geq 0$. Then 
\begin{equation*}
w(\bar{r})=(u_{2}^{\prime }(\bar{r})-u_{1}^{\prime }(\bar{r}))u_{1}(\bar{r}%
)\leq 0.
\end{equation*}%
On the other hand we have 
\begin{equation*}
V_{2}^{\prime }(r)-V_{1}^{\prime }(r)=\frac{1}{r^{d-1}}%
\int_{0}^{r}(u_{2}^{2}(s)-u_{1}^{2}(s))s^{2m+d-1}\;ds>0
\end{equation*}%
on $]0,\bar{r}]$ and therefore $V_{2}(r)>V_{1}(r)$ on $]0,\bar{r}]$. We
conclude then from the differential equation \eqref{wronskian-ode} for $w$
that $w\left( r\right) r^{2m+d-1}$ is strictly increasing on $]0,\bar{r}]$
and since $w(0)=0$ we must have $w(\bar{r})>0$ which is the desired
contradiction.
\end{proof}

\begin{rem}
\label{remark on the intervalls}From the no-crossing property stated in
lemma \ref{no-cross} it follows immediately that $\mathcal{N}_{m},$ $%
\mathcal{N}_{m,0}^{\infty }$ are intervals. More precisely, $\mathcal{N}%
_{m}=]0,a[,\mathcal{N}_{m,0}^{\infty }=]b,\infty \lbrack $ with $0\leq a\leq
b\leq \infty \lbrack $. Uniqueness of ground states is then equivalent to $%
a=b$.
\end{rem}

The important conclusion from lemma \ref{no-cross} is that two different
ground state solutions cannot intersect. From the differential equation %
\eqref{wronskian-ode} for their Wronskian $w(r)$ we see that $w(r)r^{2m+d-1}$
is a nonnegative strictly increasing function. However, we shall prove in
the sequel that $w(r)r^{2m+d-1}$ vanishes at infinity which yields the
desired contradiction. Therefore we have to analyze the decay properties of
ground states at infinity.

\begin{lem}
\label{G-riccati-decay} Let $u_{0}\in \mathcal{N}_{m,0}^{0}$. Then

\begin{equation*}
\underset{r\longrightarrow \infty }{\lim }-\frac{u^{\prime }}{u}V^{-\frac{1}{%
2}}=1.
\end{equation*}%
Consequently, for any $\kappa \in \left( 0,1\right) $, 
\begin{equation*}
\underset{r\longrightarrow \infty }{\lim \sup }\;u(r)e^{\kappa
\int_{0}^{r}V^{1/2}\left( s\right) ds}<\infty .
\end{equation*}
\end{lem}

\begin{proof}
We consider the function $z:=-\frac{u^{\prime }}{u}V^{-\frac{1}{2}}$ which
is well defined for all $r>0$ and satisfies the differential equation 
\begin{equation*}
z^{\prime }=\left( z^{2}-1\right) V^{1/2}-z\left( \frac{2m+d-1}{r}+\frac{%
V^{\prime }}{2V}\right) \;+V^{-1/2}.
\end{equation*}%
We also consider $y:=2\left( m+1\right) V-rV^{\prime }$ which satisfies the
following differential equation%
\begin{equation*}
y^{\prime }=\left( 2m+d\right) V^{\prime }-u^{2}r^{2m+1}.
\end{equation*}%
Since $u$ is decreasing, we have from the equation (\ref{int-eq})%
\begin{equation*}
V^{\prime }\left( r\right) =\frac{1}{r^{d-1}}\int_{0}^{r}u\left( s\right)
^{2}s^{2m+d-1}ds\geq \frac{u\left( r\right) ^{2}}{2m+d}r^{2m+1}.
\end{equation*}%
Hence $y$ is increasing and we have the upper bound 
\begin{equation*}
\frac{V^{\prime }}{V}\leq \frac{2\left( m+1\right) }{r}.
\end{equation*}%
Now choose $\tilde{r}$ such that $\frac{3m+d}{r}V^{-1/2}\leq 1/2$ for all $%
r\geq \tilde{r}$. Consider the direction field in the $(r,z)$ plane for the
preceding differential equation. In the set $r\geq \tilde{r}$, $z\geq 2$ we
have 
\begin{eqnarray*}
z^{\prime } &\geq &\left( z^{2}-1\right) V^{1/2}-z\frac{V^{1/2}}{2}+V^{-1/2}
\\
&\geq &\frac{1}{2}z^{2}V^{1/2}+\frac{1}{2}\left( z+1\right) \left(
z-2\right) V^{1/2} \\
&\geq &\frac{1}{2}z^{2}V^{1/2}\left( \tilde{r}\right)
\end{eqnarray*}%
It follows that, should $z(r)$ ever enter this region, it would blow up at
finite time after $\tilde{r}$ which is impossible. Hence $z$ remains
bounded. This also implies

\begin{equation*}
\underset{r\longrightarrow \infty }{\lim }u^{\prime }(r)V^{-1/2}=0.
\end{equation*}

Therefore we may apply l'Hospital's rule. We obtain 
\begin{equation*}
\underset{r\longrightarrow \infty }{\lim }z^{2}=\underset{r\longrightarrow
\infty }{\lim }\Big(\frac{u^{\prime \prime }}{u}\frac{1}{V}+\frac{1}{2}\frac{%
V^{\prime }}{V^{3/2}}z\Big)=\underset{r\longrightarrow \infty }{\lim }\Big(%
\frac{V-1}{V}-\frac{2m+d-1}{r}V^{-1/2}z\Big)=1.
\end{equation*}%
This proves the first part of the lemma.

Then for any $\kappa \in \left( 0,1\right) $ and $r$ sufficiently large, $-%
\frac{u^{\prime }}{u}V^{-1/2}\geq \kappa $ and the proof is completed by
integrating this inequality.
\end{proof}

Now we are in position to prove our uniqueness result:

\begin{thm}
\label{uniqueness-in-G}The set $\mathcal{N}_{m,0}^{0}$ has exactly one
element.
\end{thm}

\begin{proof}
Let $u_{1}(0),u_{2}(0)\in \mathcal{N}_{m,0}^{0}$ such that $%
u_{2}(0)>u_{1}(0) $. By lemma \ref{no-cross} the corresponding solutions $%
u_{1},u_{2}$ cannot intersect and we have $u_{2}(r)>u_{1}(r)>0$ for all $%
r\geq 0$. From the differential equation \eqref{wronskian-ode} for their
Wronskian $w(r)$ we see that $w(r)r^{2m+d-1}$ is a nonnegative strictly
increasing function since 
\begin{equation*}
\Big(w\left( r\right) r^{2m+d-1}\Big)^{\prime
}=(V_{2}-V_{1})u_{1}u_{2}r^{2m+d-1}>0
\end{equation*}%
and $w(0)=0$. On the other hand, we claim that 
\begin{equation*}
\underset{r\longrightarrow \infty }{\lim }\;w\left( r\right) r^{2m+d-1}=0.
\end{equation*}%
Indeed, since $u_{2}$ cannot intersect $u_{1}$ we have $\underset{%
r\rightarrow \infty }{\lim }V_{2}(r)>V_{1}\left( r\right) \geq 1$.
Trivially, $V_{2}(r)\leq \frac{u_{2}(0)^{2}r^{2\left( m+1\right) }}{2\left(
m+1\right) \left( 2m+d\right) }$. From the integral equation \eqref{int-eq}
for $u_{2}^{\prime }$, 
\begin{equation*}
u_{2}^{\prime }\left( r\right)
r^{2m+d-1}=\int_{0}^{r}(V_{2}(s)-1)u_{2}(s)s^{2m+d-1}\;ds
\end{equation*}%
and the decay properties of $u_{2}$ given in lemma \ref{G-riccati-decay} it
follows then that $u_{2}^{\prime }r^{2m+d-1}$ and $u_{2}r^{2m+d-1}$ are
uniformly bounded. Therefore 
\begin{equation*}
|w(r)r^{2m+d-1}|\leq |u_{1}||u_{2}^{\prime }r^{2m+d-1}|+|u_{1}^{\prime
}||u_{2}r^{2m+d-1}|\leq c_{1}|u_{1}|+c_{2}|u_{1}^{\prime }|
\end{equation*}%
for some positive constants $c_{1},c_{2}$ which concludes the proof.
\end{proof}


\bibliographystyle{amsplain}
\bibliography{references}

\end{document}